\documentclass{article}

\usepackage{arxiv}

\usepackage[utf8]{inputenc} 
\usepackage[T1]{fontenc}    
\usepackage{hyperref}       
\usepackage{url}            
\usepackage{booktabs}       
\usepackage{amsfonts}       
\usepackage{nicefrac}       
\usepackage{microtype}      
\usepackage{lipsum}		

\hypersetup{
	colorlinks=true, 
	citecolor=blue} 

\title{Compound Dirichlet Processes}

\author{
  Arrigo Coen\thanks{CORRESPONDING AUTHOR: Arrigo Coen, Email: coen@ciencias.unam.mx} \\
  Departamento de Matem\'aticas, Facultad de Ciencias \\
  Universidad Nacional Aut\'onoma de M\'exico\\
  M\'exico, CDMX, Apartado Postal 20-726, 01000, M\'exico\\
  \texttt{coen@ciencias.unam.mx} 
   \And
Beatriz God\'inez-Chaparro\\
Departamento de Sistemas Biol\'ogicos, Divisi\'on de Ciencias Biol\'ogicas y de la Salud,\\
Universidad Aut\'onoma Metropolitana$-$Xochimilco,\\
Mexico City, Mexico,\\ 
\texttt{bgodinez@correo.xoc.uam.mx}
}

  \usepackage{natbib}

\usepackage{amssymb,amsmath,amsthm,amsfonts,epsfig} 

\usepackage{varioref}
\usepackage{hyperref}
\usepackage{cleveref}
\usepackage{graphicx}
\usepackage{caption}
\usepackage{subcaption}  
\usepackage{tikz}
\usepackage{dsfont}

\usepackage{savesym}
\savesymbol{AND}
\usepackage{algorithm}
\usepackage{algorithmic} 
\usepackage{color}
\usepackage{soul}
\usepackage{cancel}

\usepackage{mathtools} 
\newcommand{\pr}[1]{\mathbb{P}\left[#1\right]}
\newcommand{\esp}[1]{\mathbb{E}\left[#1\right]}

\newcommand{\dist}[1]{\mathrm{#1}} 
\newtheorem{thm}{Theorem}
\newtheorem{lemma}{Lemma}

\newtheorem{prop}{Proposition}

\newtheorem{definition}{Definition}

\newcommand{\mb}{\mathbb}
\usepackage{booktabs} 
\usepackage{multirow} 

\newcommand{\alphaT}{\alpha^T}
\newcommand{\alphaX}{\alpha^X}

\newcommand{\GT}{G^T}
\newcommand{\GX}{G^X}

\newcommand{\GTo}{G_0^T}
\newcommand{\GXo}{G_0^X}

\newcommand{\FTtheta}{F_{\theta_j}^T}
\newcommand{\FXtheta}{F_{\theta_i}^X}

\newcommand{\FT}{F^T}
\newcommand{\FX}{F^X}

\newcommand{\thetaXi}{\theta_i^X}
\newcommand{\thetaTj}{\theta_j^T}

\begin{document}
\maketitle

\begin{abstract}
	The compound  Poisson process and the Dirichlet process are the pillar structures of Renewal theory and Bayesian nonparametric theory, respectively. Both processes have many useful extensions to fulfill the practitioners needs to model the particularities of data structures. Accordingly, in this contribution we joined their primal ideas to construct the compound Dirichlet process and the compound Dirichlet process mixture. As a consequence, these new processes had a fruitful structure to model the time occurrence among events, with also a flexible structure on the arrival variables. These models have a direct Bayesian interpretation of their posterior estimators and are easy to implement. We obtain expressions of posterior distribution, nonconditional distribution and expected values. In particular to find these formulas we analyze sums of random variables with Dirichlet process priors. We assessed our approach by applying our model on a real data example of a contagious zoonotic disease. 
\end{abstract}


\section{Introduction}


In this contribution are presented two continuous time processes that are probabilistically constructed through a random sum, using the framework of Bayesian nonparametric models. As a consequence of its construction, these processes could be use to model renewal phenomena. Examples of applied Bayesian nonparametric models to analyze Renewal theory phenomena are presented in  \cite{Bulla2007,Frees1986,Xiao2015}. One of the principal reasons to combine these methodologies is the fact that in many cases the renewal phenomena have complex random structures. Therefore, for these type of analysis could be better to let the data to speak by itself. By using parameter-free models, important hidden structures unveil, whereas a parametric model may conceal them. Although the combination of these branches is not new, the use of Dirichlet process that is here presented is a novel technique. 

For many applied statisticians random sums models are everyday tools. An advantage of these models is that they allow us to examine the data as the contribution of simpler parts, which improve calculations and predictions. To choose a random sum model there are three key probability concepts to have in mind: 1) the law governing the number of terms to add; 2) the dependence among the terms, and; 3) the interactions between 1) and 2). For instance, in  \cite{Rolski_book} these concepts are applied to model the behavior of insurance claims by taking in to account: 1) how many insurance claims are received in a fixed period of time; 2) the dependence of claims sizes, and; 3) the connections between the number of claims and their sizes (see also, \cite{Gebizlioglu2018,Coen2015}). Other fields where random sum models are currently applied are Multivariate analysis to model daily stock values \cite{Nadarajaha2019}, Bayesian nonparametric theory to estimate the total number of species in a population \cite{Zhang2005},  and Finance to estimate the skewed behavior of a time series \cite{Nadarajah2017a}.

The classical theory of renewal processes focus on the analysis of counting process where the interarrival times are independent and identically distributed (i.i.d.). The most remarkable example of renewal process is the Poisson process, whose interarrival times are i.i.d. exponential variables \cite{kingman1992poisson}. By allowing some interaction among the variables, this model has been generalized to resemble more intricate phenomena. Examples of these generalizations are the Cox process, the non-homogeneous Poisson process and the Markov and Semi-Markov renewal models. A thorough analysis of these models is presented in \cite{kovalenko1997}.

To define our model, we will use one of the most influential Bayesian nonparametric structure, the Dirichlet process (DP) prior \cite{Ferguson1973}. The DP effectiveness is exhibit by its successful application in many statistical analysis. As pointed out by Ferguson in \cite{Ferguson1973},  two desirable properties of a prior distribution for nonparametric problems are: a large support and a manageable posterior distribution. The DP prior handles both properties in a remarkable manner, with the clear interpretation of its parameters. Moreover, the many representations of the Dirichlet process had rise diverse important Bayesian nonparametric contributions: neutral to the right processes \cite{Doksum1974}, normalized log-Gaussian processes \cite{Lenk1988}, stick breaking priors \cite{Lijoi2005,Ishwaran2001},  species sampling models \cite{Pitman1996}, Poisson-Kingman models \cite{Pitman2003} and normalized random measures with independent increments \cite{Prnster2003}, to mention a few. Each of these models generalize an aspect of the Dirichlet process in some direction, thus, obtaining more modeling flexibility with respect to some specific feature of the data.

In this study we are applying the Dirichlet process as a mechanism to control the probability structure of a random sum stochastic process. Under this framework, we inherit the flexibility of the DP to resemble the data behavior and have a wide spectrum of probability structures to establish as prior believes. Also, we gain an interpretation of the clustering structure of the renewals and an efficient posterior simulation algorithm. In fact, these models allow us to analyze the cluster behavior of the time and space components, induced by the discrete random measures settle to each component.

\section{Compound Dirichlet Processes}

In this section we define the stochastic structure of the compound Dirichlet process  and the compound Dirichlet process mixture, and show some of their appealing modeling properties. These processes could be applied to phenomena where the a stochastic-time component defines the arrivals of random variables. Under this framework, we settle a dependence structure among arrivals and another among the events of the arrivals, keeping independence between each other.  

Let us first consider a sequence of positive random variables $ \{T_j\}_{j=1}^\infty $, and define its renewal process $ \{N_t\}_{t\in\mb{R}_+} $  as
\[ N_t =\sup\{j\in\mb{N}:T_1+T_2+\cdots+T_j< t \}, \quad t\in\mb{R}_+,\]
where the random variables $ T_j $ have the interpretation as the interarrival times between events of the phenomenon of study. Then, $N_t $ is the number of events that take place before time $ t $.  The general theory of exchangeable renewal models is studied in \cite{Coen2018}, however, here we analyze the particular implications of the DP prior framework. To this end, similarly to the ideas of a compound Poisson process, we will focus our analysis on the random process $ \{S_t\}_{t\in\mb{R}_+} $ given by
\begin{equation}\label{eq_sum_St}
S_t=\sum_{i=1}^{ N_t}X_i, \qquad t\in\mb{R}_+,
\end{equation}
where $ \{X_i\}_{i=1}^\infty$ are independent of  $ \{T_j\}_{j=1}^\infty $ . In this construction we will place two exchangeable structures: one over the  events $ \{X_i\}_{i=1}^\infty$ and one on their inter-arrival times  $ \{T_j\}_{j=1}^\infty $. The advantage of assuming this symmetric structure lies in the fact that with it we could model various dependence behaviors and, at the same time, allows the analysis of cluster formations among variables \cite{Aldous1985}. To define these dependent structures we will use Dirichlet process priors. The DP prior model is define as
\begin{align}\label{eq_general_DP}
X_i\mid G &\sim  G,\\
G &\sim \mbox{DP}(\alpha, G_0),\nonumber
\end{align}
where $\mbox{DP}(\alpha, G_0)$ denotes a Dirichlet process with precision parameter $\alpha>0$ and base distribution $G_0$. The DP random measure $ G $ is define in \cite{Ferguson1973} by the distributional property 
\[ (G(A_1),\ldots,G(A_k))\sim \dist{Dir}(\alpha G_0(A_1),\ldots,\alpha G_0(A_k)), \]
for all measurable partition $(A_1,\ldots,A_k) $ of the sample space of $ G_0 $, where $ \dist{Dir}(a_1,\ldots,a_k)$ denotes the Dirichlet distribution of $k$-dimension  with parameter $ (a_1,\ldots,a_k) $. An implication of these assumptions is that the joint distribution of $(X_{1},\ldots,X_{n})$  can be factorized using the generalized P\'olya urn  scheme \cite{Blackwell1973}, i.e. for any $n>1$,
\begin{equation}\label{eq_Polya_scheme}
X_n|X_{n-1},X_{n-2},\ldots,X_1 \sim \dfrac{\alpha}{\alpha+n-1}G_0 + \dfrac{1}{\alpha+n-1}\sum_{i=1}^{n-1} \delta_{X_i},
\end{equation}
where $ \delta_x $ denotes de Dirac measure at $ x $. This last expression could be interpreted as $ X_n $ given $ X_{n-1},X_{n-2},\ldots,X_1 $ has probability $ \frac{\alpha}{\alpha+n-1} $ of being a new $G_0 $-distributed random variable independent of the past values and probability $ \frac{n-1}{\alpha+n-1} $ to repeat a previously seen value. This also implies that the random variables $ \{X_i\}_{i=1}^\infty$ are exchangeable, meaning that the joint distribution of $(X_{1},\ldots,X_{n})$ is equal to the distribution of $(X_{\pi_1},\ldots,X_{\pi_n})$, for any permutation $ \pi $ of $ \{1,\ldots,n\} $ \cite{Aldous1985}. It follows that the variables $ X_i $ are conditionally independent and identically distributed $G_0$, with constant correlation given by
\[ {\mathsf{Corr} (X_i,X_j)}= \frac{1}{\alpha+1},  \qquad  i,j\in\mb{ N}. \] 

A little difficulty to work directly with DP priors is that their samples are almost sure discrete random measures. Many works overcome this difficult by using a DP as a prior over the distribution of an extra layer of parameters \cite{Ferguson1983,Lo1984,Escobar1995}. In fact, in many cases these parameters help to make the description simpler and have a direct interpretation. These models are known as the Dirichlet process mixtures models, and they are define by the structure
\begin{align}\label{eq_general_DPM}
X_i\mid \theta_i&\sim F_{\theta_i}\nonumber\\
\theta_i\mid G &\sim  G,\nonumber\\
G &\sim \mbox{DP}(\alpha, G_0),\nonumber
\end{align}
where $ F_\theta $ denotes a member of a fixed family of distributions parametrized  by $ \theta $. Even thought, this last approach add a hidden extra layer of parameters, there are many Gibbs sampling methods to confront this issue \cite{Ishwaran2001,Maceachern1998,Neal2000}. Furthermore, the discreteness of the random measures of DP allows to study the clustering properties of the data \cite{Escobar1994a,Gorur2010,Escobar1995a}. Under these notations we can establish a nonparametric structure on \eqref{eq_sum_St}.
\begin{definition}\label{def_CDP}
	A continuous time stochastic process $ \{S_t\}_{t\in\mb{R}_+} $ given by \eqref{eq_sum_St} is a compound Dirichlet process (CDP) if it follows the stochastic structure
	\begin{align*}
	T_j\mid \GT  &\sim  \GT ,& X_i\mid \GX  &\sim  \GX ,\\
	\GT  &\sim \mbox{DP}(\alphaT, \GTo), &\GX  &\sim \mbox{DP}(\alphaX, \GXo),
	\end{align*}	
	where $ \{X_i\}_{i=1}^\infty$ is independent of  $ \{T_j\}_{j=1}^\infty $. To simplify the notation, we use  $ S_t\sim\dist{CDP}(\alphaT, \GTo,\alphaX, \GXo) $. 
\end{definition}

It is important to notice that, as in the DP framework, the CDP model also has a positive probability of repeat previously seen values. In the classical DP model, as $ n\to\infty $ the expected number of distinct $ X_i $ terms in $ \{X_1,\ldots,X_n \} $ grows as $ \alpha\log n $ \cite{Korwar1973}; it is important to notice tat this rate is smaller than $ n $. Consequently, the CDP has a positive probability of repeat increments. In other words, there is a positive probability that the increment $ S_{t_2}-S_{t_1}$ is equal to the increment $S_{t_4}-S_{t_3} $, for any positive real numbers $ t_1<t_2 $ and $ t_3<t_4 $. Nevertheless, it is important to notices that the rate of  repeated values is even smaller than the one of the DP framework. The addition operation confers a decrease in the number of repeated values; selecting different adding terms gives an extra possibility of different total results. In order to diminish the problem of repeated values and to study the clustering structure of the random variables, we have the next definition.
\begin{definition}\label{def_CDPM}
	A continuous time stochastic process $ \{S_t\}_{t\in\mb{R}_+} $ given by \eqref{eq_sum_St} is a compound Dirichlet process mixture (CDPM) if it follows the stochastic structure
	\begin{align*}
	T_j\mid \thetaTj &\sim \FTtheta & X_i\mid \thetaXi&\sim \FXtheta \\
	\thetaTj\mid \GT  &\sim  \GT ,& \thetaXi\mid \GX &\sim  \GX,\\
	\GT &\sim \mbox{DP}(\alphaT, \GTo), &\GX &\sim \mbox{DP}(\alphaX, \GXo),
	\end{align*}
	where $ \{X_i\}_{i=1}^\infty$ is independent of  $ \{T_j\}_{j=1}^\infty $. We use  $ S_t\sim\dist{CDPM}(\alphaT, \GTo,\FT,\alphaX, \GXo,\FX) $ to denote this process, where $ \FT $ and $ \FX $ represent parametric families of distributions.
\end{definition}

Under definitions \ref{def_CDP} and  \ref{def_CDPM} we have a fruitful structure to consider the time evolution of the accumulation of random variables. For instance, we could use this models to analyze the claims process of an insurance company; we could establish the next methodology to analyze the Capital Requirements of the company at year $ t $ \cite{Linder2004}. First, we use the data of years $ t-5 $ and $ t-4 $ to fit the base distributions $ \GTo $ and $ \GXo $. Then, we use the data of years $ t- 3$, $ t- 2$ and $ t- 1$ as a sample to obtain the posterior distribution of the CDPM model.  Finally, under Bayesian methodologies we obtain point estimators, confidence intervals, and also made hypothesis testing over the Capital Requirements. To satisfy these and other inference inquires,  the next section presents statistical implications of the CDP and CDPM models.

\subsection{Some properties and results of CDP and CDPM}

Let us continue with some properties of the CDP and CDPM models. These results are exhibit under the CDP framework, however their implication on the CDPM models are direct. The results are arrange in order to calculate, or at least  approximate, the posterior distribution of $ S_t $.
\begin{thm}\label{theo_dist_S_t}
	If $ S_t\sim\dist{CDP}(\alphaT, \GTo,\alphaX, \GXo) $, then for any $ t,s\in\mb{R}_+ $
	\begin{equation}\label{eq_dist_S_t}
	\pr{S_t\le s} = \sum_{n=1}^\infty\sum_{v\in\Delta_n} p_v(n)\left( H_1^{*v_1}\ast H_2^{*v_1}\ast\cdots\ast H_n^{*v_n}\right) (s)\pr{N_t=n}, 
	\end{equation}
	where  $ \Delta_n= \left\{v=(v_1,\ldots,v_n) \in\mb{ N}^n: \sum_{i=1}^niv_i=n \right\}$,   
	\[ p_v(n)= \dfrac{n!}{\alphaX(\alphaX+1)\cdots(\alphaX+n-1)} \prod_{i=1}^n \dfrac{(\alphaX)^{v_i}}{i^{v_i}v_i!},   \]
	$H_i^{*v_i}$ is the $ v_i $-convolution of the distribution $H_i(\cdot)=\GXo(\cdot/i) $, and $ \left( H_1^{*v_1}\ast H_2^{*v_1} \ast \cdots \ast H_n^{*v_n}\right) (\cdot) $ is the convolution of these convolutions.
\end{thm}

The proof of Theorem \ref{theo_dist_S_t} is a direct consequence of the next lemma.

\begin{lemma}\label{prop:nonparametric1}
	If $ X_i| G \sim  G $ and $ G \sim \mbox{DP}(\alpha, G_0) $, we define $ \{S_n\}_{n=1}^\infty $ by
	\[ S_n=\sum_{i=1}^n X_i, \qquad n\in\mb{ N}.\]
	Then 
	\[\pr{S_n\le s} = \dfrac{	n!}{\alpha(\alpha+1)\cdots(\alpha+n-1)}\sum_{v\in\Delta_n}\left( H_1^{*v_1}\ast H_2^{*v_1}\ast\cdots\ast H_n^{*v_n}\right) (s) \prod_{j=1}^n \dfrac{\alpha^{v_j}}{j^{v_j}v_j!} ,\]
	for every $ s\in\mb{R}_+ $,	where $H_i^{*v_i}$ is the $ v_i $-convolution of the distribution $H_i(\cdot)=\GXo(\cdot/i) $. 
	Moreover, if we define 
	\[ N_t^X =\sup\{n\in\mb{N}:X_1+X_2+\cdots+X_n< t \}, \quad t\in\mb{R}_+.\]
	then $\pr{N_t^X=n}= \pr{S_n\le t}$ for  $ t>0 $.
\end{lemma}

\begin{proof} 
	For the sake of completeness, we present the proof presented in \cite{Coen2018} for this result. Let $V=(V_1,\ldots,V_n)$  be the random vector indicating the repeated values in $(X_1,\ldots,X_n)$, under the following scheme: there are $V_1$ values that only repeats once, $V_2$ values that repeats twice, and so on. 	
	Then, conditioning on $V$ the distribution of $ S_n $ can be written as 
	\[ \pr{S_n\le s} = \sum_{v\in\Delta_n }\pr{X_1+\ldots+X_n\le s|V=v}\pr{V=v}. \]
	The conditional distribution of $ X_1+\ldots+X_n $ given $ V=v $ is equal to the convolution of $ v_1 $ independent variables with distribution $ H_{1} $, convolved with the convolution of $ v_2 $ variables distributed $ H_2 $, and so on. We condition on $ V $ because this eliminates the repeated values of $(X_1,\ldots,X_n)$, which allow us to consider the convolution of independent variables. Consequently, we define $ H_j $ because given the repeated values of $ X_i $s we need to consider the probabilities $ \pr{jX\le t} $, for $ X\sim G_0 $ and $ j\in\mb{ N} $. Thus, we obtain
	\[ \pr{X_1+\ldots+X_n\le s|V=v} =\left( H_1^{*v_1}\ast H_2^{*v_1}\ast\cdots\ast H_n^{*v_n}\right) (s). \]
	The probabilities of $ \{V=v\} $ are given by the Ewen's sampling formula  \cite{Ewens1972}, as
	\begin{equation}\label{eq_Pr_V_v}
	\pr{V=v} = \dfrac{n!}{\alpha(\alpha+1)\cdots(\alpha+n-1)}\prod_{j=1}^n \dfrac{\alpha^{v_j}}{j^{v_j}v_j!}, 
	\end{equation}
	by applying the generalized P\'olya urn  scheme \eqref{eq_Polya_scheme} over the possible different values of $ X_1,\ldots,X_n $. Finaly, the equality $\pr{N_t^X=n}= \pr{S_n\le t}$ is a direct consequence of the definition of $ N_t^X $.
\end{proof} 

Accordingly to the last results, the distribution of the CDP could be express as an infinite sum. Although we are not presenting directly the distributions of $ S_t $ and $  N_t $, they can be express on terms of the distribution of $ S_n $, using the second statement of Lemma \ref{prop:nonparametric1}. Since Dirichlet process tends to focus most of its mass in a few atoms the convergence of the series of \eqref{eq_dist_S_t} is fast. This allow us to approximate the distribution of $ S_t $ in two ways. We can truncate the sum \eqref{eq_dist_S_t} to a finite fixed number of terms, or we can fix a quantity $ \epsilon_0 $ to count only terms with $ \pr{N_t=n}>\epsilon_0 $. In both cases we are restrain the error of the approximation. Also, our computational experiments show that both approximations are stable.

\begin{prop}\label{prop_moments}
	Given $ S_t\sim\dist{CDP}(\alphaT, \GTo,\alphaX, \GXo) $, let $ \mu _{X,i}= \esp{X_1^i}$ and $  \mu _{ N_t,i}=\esp{ N_t^i}$, for $ i=1,2,3 $, then
	\begin{align*}
	\esp{S_t}=&\mu _{N_t,1} \mu _{X,1}\\[3ex]
	\esp{S_t^2}=&\frac{\left(\mu _{N_t,2}-\mu _{N_t,1}\right) (\alpha  \mu _{X,1}^2+\mu _{X,2})}{\alpha +1}+\mu_{N_t,1} \mu_{X,2}\\[3ex]
	\esp{S_t^3}=&\frac{\left(2 \mu _{N_t,1}-3 \mu _{N_t,2}+\mu _{N_t,3}\right) (\alpha ^2 \mu _{X,1}^3+3 \alpha  \mu _{X,2} \mu _{X,1}+2 \mu _{X,3})}{(\alpha +1) (\alpha +2)}\\
	&+\frac{\left(\mu _{N_t,2}-\mu _{N_t,1}\right) \left(\alpha  \mu _{X,1} \mu _{X,2}+\mu _{X,3}\right)}{\alpha +1}+\mu _{N_t,1} \mu _{X,3}
	\end{align*}
\end{prop}

\begin{proof}
	The expression for $ \esp{S_t} $ follows conditioning on $ \{N_t=n\} $ and using the lineality of the expectation operation. To obtain the expression for $ \esp{S_t^2} $  one must consider the possible repeated values of the exchangeable sequence $ \{X_i\} $. From now, let us assume that $ \GXo $ is a continuous distribution. Then, conditioning on $   \{N_t=n\}  $, we obtain
	\begin{align*}
	\esp{S_t^2| N_t=n} =& 	\esp{\left(\sum_{i=1}^nX_i\right)^2}\\
	=& \sum_{i=1}^n\esp{X_i^2} +\sum_{1\le i<j\le n}\esp{X_i X_j}\\
	=& n \mu_{X,2} + n(n+1)\esp{X_1 X_2}\\
	=& n \mu_{X,2} + n(n+1)\left(\mu_{X,1}^2 \dfrac{\alpha}{\alpha+1}+\mu_{X,2}\dfrac{1}{\alpha+1}  \right).
	\end{align*}
	The last equality follows from conditioning to $ \{X_1=X_2\} $, and that $ \pr{X_1=X_2}= \frac{\alpha}{\alpha+1} $. This last expression gives the result for 	$ \esp{S_t^2} $. Likewise, the result for $ 	\esp{S_t^3} $ is obtained by conditioning to the possible repetitions of $ \{X_1,X_2,X_3\} $, and applying $ \pr{X_1=X_2=X_3}=\frac{1}{\alpha+1}\frac{2}{\alpha+2} $, $ \pr{X_1=X_2\ne X_3}=\frac{1}{\alpha+1}\frac{\alpha}{\alpha+2} $ and $ \pr{X_1\ne X_2\ne X_3\ne X_1}=\frac{\alpha}{\alpha+1}\frac{\alpha}{\alpha+2} $. Finally, in the discrete case  we only need to assured that we are conditioning only on cases when the variables are equal as a consequence of the  P\'olya urn's repetitions, and the formulas follow. 
\end{proof}

\begin{prop}\label{prop_M_S_n}
	Under the notation of Lemma \ref{prop:nonparametric1}, the moment generator function of $S_n  $ is given by
	\begin{equation}\label{eq_prop_M_S_n}
	M_{S_n}(t)= \dfrac{	n!}{\alpha(\alpha+1)\cdots(\alpha+n-1)}\sum_{v\in\Delta_n} \prod_{j=1}^n \dfrac{1}{v_j!}\left[\dfrac{\alpha M_X(tj)}{j}\right]^{v_j},
	\end{equation}
	where  $ M_X$ denotes the moment generator function of $ X_1 $.
\end{prop}

\begin{proof}
	Conditioning over the possible partitions we obtain
	\begin{align*}
	M_{S_n}(t)&=\sum_{v\in\Delta_n} \pr{V=v}\esp{e^{tS_n}|V=v}\\
	&=\sum_{v\in\Delta_n} \pr{V=v} \prod_{j=1}^nM_X(tj)^{v_j},
	\end{align*}
	where the last equality follows since the expected value of $ e^{tS_n} $  conditioned on $ \{V=v\} $ is the product of independent random variables equal in distribution to $ e^{tjX_1} $, each repeated $ v_j $ times for $ j=1,\ldots,n $. This gives \eqref{eq_prop_M_S_n} when applying \eqref{eq_Pr_V_v}. 
\end{proof}

To see an application of \eqref{eq_prop_M_S_n} let us consider the Gaussian  distribution case. For this base distribution, we obtain
\begin{align*}
M_{S_n}(t)&= \dfrac{n!}{\alpha(\alpha+1)\cdots(\alpha+n-1)}\sum_{v\in\Delta_n} \prod_{j=1}^n \dfrac{1}{v_j!}\dfrac{\alpha^{v_j}}{j^{v_j}} e^{t(\mu jv_j) + t^2(\sigma^2j^2v_j)/2 }\\
&= \dfrac{n!}{\alpha(\alpha+1)\cdots(\alpha+n-1)}\sum_{v\in\Delta_n} e^{t(n\mu) + t^2(\sigma^2\sum_{j=1}^nj^2v_j)/2 } \prod_{j=1}^n \dfrac{1}{v_j!}\dfrac{\alpha^{v_j}}{j^{v_j}}, \\
\end{align*}
thus, we obtain that the sum of variables with prior DP and base measure Gaussian, is the mixture of Gaussian random variables. The next result shows that the CDP is a conjugate model.

\begin{prop}
	If $ (X_1,T_1),\ldots,(X_n,T_n) $ is a random sample of $ S_t\sim\dist{CDP}(\alphaT, \GTo,\alphaX, \GXo) $, then 
	\begin{align}
	S_t&\mid (X_1,T_1),\ldots,(X_n,T_n) \nonumber \\ 
	&\sim\dist{CDP}\left(\alphaT+n, \dfrac{\alpha}{\alpha+n}\GTo+\dfrac{\alpha}{\alpha+n}\sum_{j=1}^n\delta_{T_j},\alphaX+n, \dfrac{\alpha}{\alpha+n}\GXo+\dfrac{\alpha}{\alpha+n}\sum_{i=1}^n\delta_{X_i}\right). \label{eq_post_CDP}
	\end{align} 
\end{prop}

\begin{proof}
	The proof is immediate by applying the conjugate property of the Dirichlet process prior and the independence of $ \{T_j\}_{j=1}^\infty $ with $ \{X_i\}_{i=1}^\infty$. 
\end{proof}

\subsection{Two examples of flexible base mesures for CDP and CDPM}

Let us continue by presenting two examples of distribution families for the base distribution $\GXo$ where the convolution of \eqref{eq_dist_S_t} simplifies. These families are the Gaussian and the phase-type distribution. Is important to notices that both had a wide support that allow to approximate other distributions. First, in the case of Gaussian distributions defined by $\GXo= \dist{ N}(\mu,\sigma^2) $, we obtain that $ H_j = \dist{ N}(\mu j,\sigma^2j^2) $, and so
\[ H_1^{*v_1}\ast H_2^{*v_1}\ast\cdots\ast H_n^{*v_n}  = \dist{ N}\left(\mu n,  \sigma^2\sum_{j=1}^nj^2v_j\right). \]
This implies that the density of $ S_n $ is given by
\[ f_{S_n}(t)= \dfrac{n!}{\alpha(\alpha+1)\cdots(\alpha+n-1)}\sum_{v\in\Delta_n}\dfrac{e^{-(t-\mu n)^2/2\sigma^2\sum_{j=1}^nj^2v_j}}{\sqrt{2\pi \sigma^2\sum_{j=1}^nj^2v_j }} \prod_{j=1}^n \dfrac{\alpha^{v_j}}{j^{v_j}v_j!} .\]
Rates of the convergence of Gaussian mixtures to the true underlying distribution are presented in \cite{Ghosal2007,Tokdar2006}. As a consequence of this convergence we could use the Gaussian model in cases with poor prior information.

For the second example we present the analytic expression for $ f_{S_n} $ is the case of phase-type distributions. An excellent account of the theory of phase-type and matrix-exponential distributions is presented in \cite{Bladt2017}. An important property of this family is that it is dense on the set of positive random variables; i.e., any positive random variable can be arbitrary approximated by a phase-type distribution. We will denote by $ U\sim \dist{PH}_p(\mathbf{\pi},\textbf{T}) $, a random variable with phase-type density given by
\[ f(u) = \mathbf{\pi} e^{\textbf{T}u}\textbf{t}, \]
where $ \mathbf{\pi} = (\pi_1,\pi_2,\ldots,\pi_p) $ is a probability row vector, $ \textbf{T} $ a subgenerator matrix of dimension $ p $, and $ \textbf{t}=-\textbf{T}\textbf{1} $, with $ \textbf{1} $ the vertical vector of ones of length $ p $. Then, under this notation, if $ \GXo=\dist{PH}_p(\mathbf{\pi},\textbf{T})$, we obtain that $ H_j = \dist{PH}_p(\mathbf{\pi},\textbf{T}/j) $. By applying the convolution property of phase-type variables:
\[ Z_1+Z_2\sim \dist{PH}_{p_1+p_2}\left((\mathbf{\pi}_1,0), \begin{bmatrix}
\textbf{T}_1 &  \textbf{t}_1 \mathbf{\pi}_2\\
0& \textbf{T}_2
\end{bmatrix} \right),  \]
for $ Z_1\sim \dist{PH}_{p_1}(\mathbf{\pi}_1,\textbf{T}_1) $ and $ Z_2\sim \dist{PH}_{p_2}(\mathbf{\pi}_2,\textbf{T}_2) $, with $ \textbf{t}_1=-\textbf{T}_1\textbf{1} $. Thus, 
\[ H_j^{*v_j} = \dist{PH}_{pv_j}((\mathbf{\pi},0,\ldots,0),\textbf{T}(j,v_j)),\]
where $ \textbf{T}(j,v_j) $ is a matrix of dimension $ pv_j\times pv_j $, given by
\[ \textbf{T}(j,v_j)= \begin{bmatrix}
\textbf{T} & -\textbf{T}\textbf{1}\mathbf{\pi} & 0 & \dots  & 0 \\
0& \textbf{T} & -\textbf{T}\textbf{1}\mathbf{\pi} &  \dots  & 0 \\
0&0& \textbf{T} &  \dots  & 0 \\
0&0&0 & \dots  & 0 \\
\vdots & \vdots & \vdots & \ddots & \vdots \\
0&0&0&0  & \textbf{T}
\end{bmatrix} / j.\]
This implies
\[H_1^{*v_1}\ast H_2^{*v_1}\ast\cdots\ast H_n^{*v_n} = \dist{PH}_{p\sum v_j}\left((\mathbf{\pi},0,\ldots,0),\textbf{T}^v\right),\]
where $ \textbf{T}^v $ is a matrix of dimension $ p\sum v_j\times p\sum v_j $, given by
\[ \textbf{T}^v= \begin{bmatrix}
\textbf{T}(1,v_1) & -\textbf{T}\textbf{1}\mathbf{\pi} & 0 & 0 & \dots  & 0 \\
0& \textbf{T}(2,v_2) & -\textbf{T}\textbf{1}\mathbf{\pi}/2  & 0&  \dots  & 0 \\
0&0& \textbf{T}(3,v_3)  & -\textbf{T}\textbf{1}\mathbf{\pi}/3&  \dots  & 0 \\
0&0&0  & \textbf{T}(4,v_4)  & \dots  & 0 \\
\vdots & \vdots& \vdots & \vdots & \ddots & \vdots \\
0&0&0&0 & 0  & \textbf{T}(n,v_n)
\end{bmatrix}. \]
We eliminate from $ \textbf{T}^v $ the rows where $ v_j=0 $.  This implies that the density of $ S_n $ is given by
\[ f_{S_n}(u) = \dfrac{n!}{\alpha(\alpha+1)\cdots(\alpha+n-1)}\sum_{v\in\Delta_n} \mathbf{\pi} e^{\textbf{T}^vu}\mathbf{t}^v \prod_{j=1}^n \dfrac{\alpha^{v_j}}{j^{v_j}v_j!}, \]
where $ \mathbf{t}^v = -\textbf{T}^v\textbf{1} $. Thus, the sum of variables with DP prior and phase-type base distribution is the mixture of phase-type random variables.

\section{An application to dog rabies}


Rabies is one of the most severe zoonotic diseases. It is caused by a rhabdovirus in the genus Lyssavirus and infects many mammalian species. It can be transmitted through infected saliva, and it is almost fatal following the onset of clinical symptoms \cite{webber2009communicable}. Up to 99\% of cases, domestic dogs are responsible for rabies virus transmission to humans. In Africa, an estimated 21,476 human deaths occur each year due to dog-mediated rabies, 36.4\% of global human deaths \cite{WorldHealthOrganization2018}. To have effective interventionism against zoonotic infections, it is important to recognize whether infected individuals stem from a single outbreak sustained by local transmission, or from repeated introductions \cite{Cori2018,Ypma2013}. 

Probability models that are commonly apply to epidemiological spreads are: coupling structures,  random graphs, EM-algorithm and MCMC methods. For a recent account of the theory we refer the reader to \cite{Andersson2000} and \cite{Brauer2008}. These models have dependence structures to resemble the infectious rate as a function of infected individuals in its vicinity. Another important quality of these models is to admit censored data, since often the epidemic process is only partly observed. These two properties are also found in the CDPM model. The spacial vicinity is handle directly by the posterior distribution; areas where the cases of the disease are found have a bigger probability of new cases. Censored data can be manage if the censore is in time or in the spacial component.

In \cite{Cori2018} is analyze data on dog rabies reported in Bangui, the capital of the Central African Republic between the 6th January 2003 and the 6th march 2012, they report 151 rabies cases, with information on report date, spatial coordinates, and genetic sequence of the isolated virus. The data are available in the package \texttt{outbreaks} of R. They applied a clustering graph model for each component and extract the most connected dots by pruning. We study this data using the CDPM model under the following assumptions. To model the time component we use an exponential mixture kernel  with  $ \dist{Ga(\alpha_0,\beta_0)} $ base distribution, and to model the spatial component we use a Gaussian mixture kernel  with  Gaussian-Inv-Gamma base distribution. To fit these Dirichlet process mixtures we use the Gibs sampling methodology.

\begin{figure}[h!]
	\centering
	\includegraphics[scale=.5]{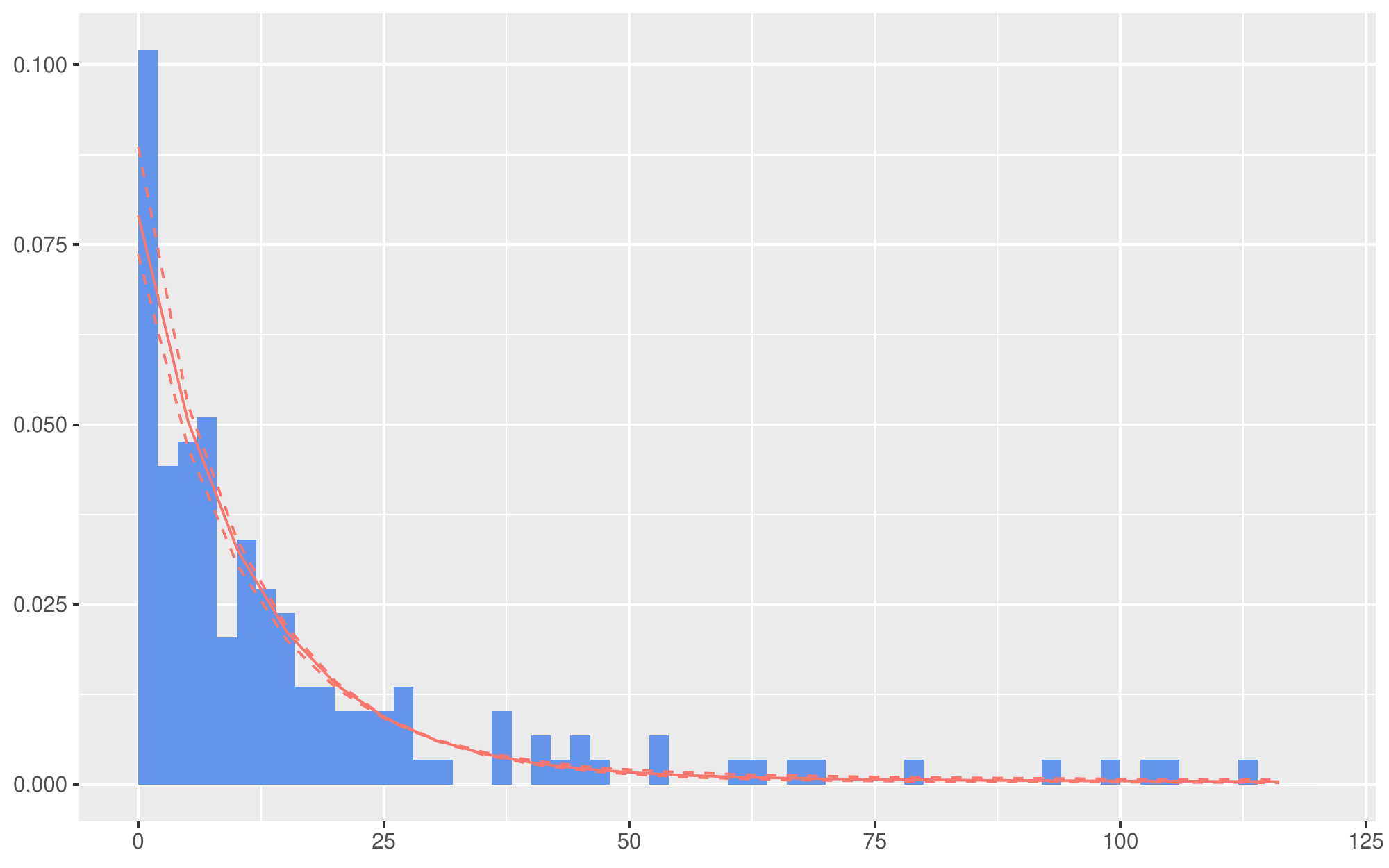}
	\caption{Visualization of the density estimation of the rabies data using the CDPM model. In this figure are exhibit the density estimation of an exponential mixture kernel  with  $ \dist{Ga(\alpha_0,\beta_0)} $ base distribution. The prior parameters of the Gamma base measure are  $ \alpha_0=1 $ and $ \beta_0=8 $.  In this figure is compared the data histogram against the posterior density estimator with its 95\% confidence interval.  }
	\label{fig_dens_T_i}       
\end{figure}

Figure \ref{fig_dens_T_i} shows that the model is able to capture
the density pattern of the time component. In this figure is compared the data histogram against the posterior density estimator with its 95\% confidence interval. As pointed by \cite{Cori2018}, the dates of the reports are close. This is characterized by the appearance of only two mixture components in the posterior distributions. Figure \ref{fig_map_rabies} present the spatial cluster behavior of the posterior distribution. Even though some spatial components are missing, we could obtain the posterior estimators of the spatial data; this missing values does not affect our analysis since we are assuming independence between spatial and temporal variables. In comparison with the results of \cite{Cori2018} we obtain almost the same cluster structure. 

In this application, the value of applying the CDPM model is on the estimation of future rabies outbreaks. Under our framework we obtained the complete probability structure of probable future contagious. We could answer many statistical questions through simulation using the posterior distributions of the CDPM model. For instance, we can find the probability that in the next year the number of cases doubles with respect to past year numbers. Likewise, we could obtain the spatial stochastic mobility of the disease, by locating the regions of where the disease is more concentrated. Our model allow an early assessment of infectious disease outbreaks, which is fundamental to implementing timely and effective control measures.

\begin{figure}[!h]
	\centering
	\includegraphics[scale=0.7]{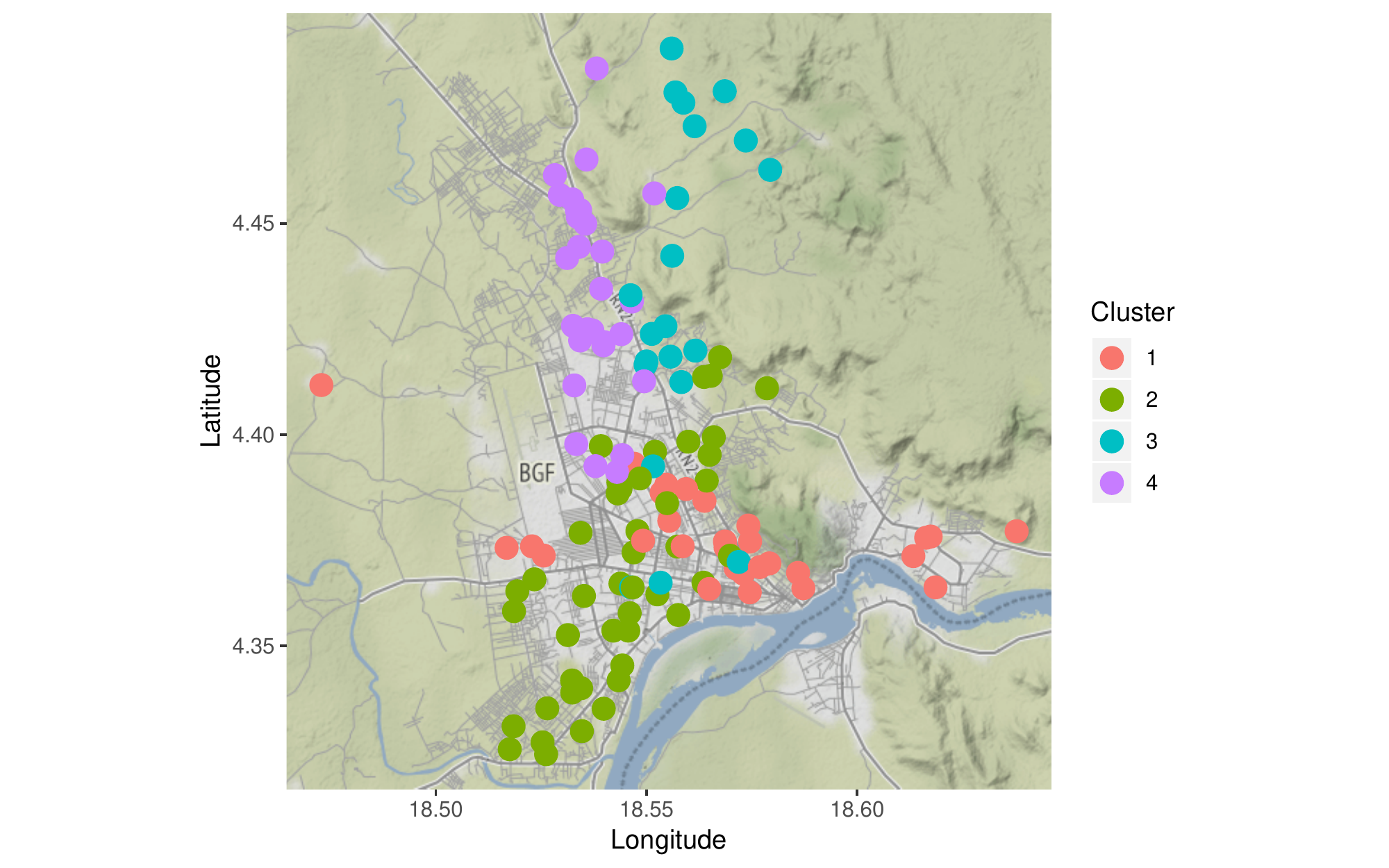}
	\caption{Visualization of spatial clusters of the rabies data using the CDPM model. In this figure are exhibit the spatial cluster structure of a Gaussian mixture kernel  with  Gaussian-Inv-Gamma base distribution.}
\label{fig_map_rabies}       
\end{figure}

\section{Discussion}

We have proposed a simple approach for statistical analysis of renewal phenomena, which combines ideas from Renewal theory  and Bayesian nonparametric theory. The proposal relies on consider two Dirichlet process priors; one on the time occurrence of event and one the arrival events. The resulting methodology is not computationally demanding and allows us to predict relatively well the evolution of renewal phenomena. Furthermore, it can be applied in cases where the cluster structure is an important factor in the analysis.

The proposed methods performs well in real spatial contexts, showing appealing features which can be used to bring closer these methodologies to practitioners in important scientific fields such as contagious analysis and general spatiotemporal analysis. Other choices of random measures, potentially lead to similar outcomes. These, more general classes priors, will be pursued elsewhere.

\section{Acknowledgments}
	The first author is grateful to Prof. Bego\~na Fern\'andez and Prof. Rams\'es Mena for their valuable suggestions on an earlier version of the manuscript. This research was partially supported by a DGAPA Posdoctoral Scholarship.

\bibliographystyle{plainnat}
\bibliography{ArXiv_CDP}

\end{document}